\documentclass[conference]{IEEEtran}
\IEEEoverridecommandlockouts
\usepackage{cite}
\usepackage{amsmath,amssymb,amsfonts}
\usepackage{amsthm}
\usepackage{algorithm, algpseudocode}  

\usepackage{graphicx}
\usepackage{textcomp}
\usepackage{xcolor}
\usepackage{tcolorbox}
\usepackage{mathtools}
\usepackage{calligra}
\usepackage{xcolor}
\usepackage{tabu}
\usepackage{caption}
\usepackage{enumitem}
\usepackage{balance}

\newlist{eqlist}{enumerate*}{1}
\makeatletter
\setlist[eqlist]{itemjoin=\,,mode=unboxed,label={$\overset{(\alph*)}{=}$},
  ref=\theequation(\alph*),before={\let\label\ltx@label}}
\makeatother

\DeclareMathAlphabet{\mathcalligra}{T1}{calligra}{m}{n}
\DeclareFontShape{T1}{calligra}{m}{n}{<->s*[2.2]callig15}{}
\newcommand{\scriptr}{\mathcalligra{r}\,}
\newcommand{\bs}[1]{\boldsymbol{#1}}

\newtheorem{lemma}{Lemma}
\newtheorem{conjecture}{Conjecture}

\newtheorem{remark}{Remark}

\newtheorem{definition}{Definition}
\newtheorem{theorem}{Theorem}

\def\BibTeX{{\rm B\kern-.05em{\sc i\kern-.025em b}\kern-.08em
    T\kern-.1667em\lower.7ex\hbox{E}\kern-.125emX}}
\begin{document}

\title{On Distributed Multi-User Secret Sharing with Multiple Secrets per User}

\author{
\IEEEauthorblockN{Rasagna Chigullapally, Harshithanjani Athi, V. Lalitha\\}
\IEEEauthorblockA{International Institute of Information Technology, Hyderabad\\
Email: \{rasagna.c@research.iiit.ac.in, \\ harshithanjani.athi@research.iiit.ac.in, lalitha.v@iiit.ac.in\}}
\and
\IEEEauthorblockN{Nikhil Karamchandani\\}
\IEEEauthorblockA{Department of Electrical Engineering, \\ Indian Institute of Technology, Bombay\\
Email: \{nikhil.karam@gmail.com\}\\}
\thanks{\hrule}%

\thanks{ Nikhil Karamchandani acknowledges support from SERB via a MATRICS grant. }

}

\renewcommand{\qedsymbol}{\ensuremath{\blacksquare}}

\maketitle

\begin{abstract}
We consider a distributed multi-user secret sharing (DMUSS) setting in which there is a dealer, $n$ storage nodes, and $m$ secrets. Each user demands a $t$-subset of $m$ secrets. Earlier work in this setting dealt with the case of $t=1$; in this work, we consider general $t$. The user downloads shares from the storage nodes based on the designed storage structure and reconstructs its secrets. We identify a necessary condition on the storage structures to ensure weak secrecy. We also make a connection between storage structures for this problem and $t$-disjunct matrices. We apply various $t$-disjunct matrix constructions in this setting and compare their performance in terms of the number of storage nodes and communication complexity. We also derive bounds on the optimal communication complexity of a distributed secret sharing protocol. Finally, we characterize the capacity region of the DMUSS problem when the storage structure is specified.
\end{abstract}

\section{Introduction}
A secret sharing scheme is a cryptographic technique that distributes a secret among multiple users while maintaining two key properties: secret recovery, which allows authorized subsets of parties to reconstruct the secret from their shares, and collusion resistance, which ensures that unauthorized subsets of parties cannot learn anything about the secret. These properties are crucial in maintaining the confidentiality and integrity of the secret.

The concept of secret sharing was first introduced by Shamir \cite{AS79} and Blakley \cite{Blakley1899} in their independent works. In \cite{AS79}, Shamir proposes a secret sharing scheme based on polynomial interpolation, while in \cite{Blakley1899}, Blakley introduces a secret sharing scheme based on the intersection of subspaces. Secret sharing has been extensively studied and applied in various areas of cryptography and distributed computing, such as secure multiparty computation \cite{cramer2000general}, 
secure cloud computing \cite{takahashi2013secret}, 
secure voting systems \cite{liu2019voting}, to name a few.
Moreover, recent advances in secret sharing have enabled its use in emerging technologies such as blockchain \cite{raman2018distributed} and secure multi-party machine learning \cite{baccarini2020multi}.

In a conventional secret-sharing scenario, the key assumption is that the dealer has a direct communication channel to all users. Hence, the encoded secret shares are readily available to the users. However, in various scenarios, such as network coding and distributed storage scenarios, the communication between the dealer and users can be mediated by intermediary nodes. Specifically, in a distributed storage scenario, the dealer stores the encoded shares in the storage nodes, and the users can access a certain subset of them. This introduces complexities where the dealer has to ensure that only the authorized users can reconstruct the designated secret. 

The scenario of distributed storage was considered in recent work by Soleymani et al. \cite{SM18}. A distributed multi-user secret sharing (DMUSS) system was considered, which consists of a dealer, $n$ storage nodes, and $m$ users. In this scenario, each user has a designated secret message and is given access to a certain subset of storage nodes, where the user can download the stored data. 
To ensure that certain privacy conditions are satisfied, the Sperner family \cite{sperner} is used in obtaining these subsets of storage nodes. The dealer is treated as a master node controlling all the storage nodes. The dealer aims to securely share a specific secret $s_j$ with user $j$ via the storage nodes. Under the multi-user context, two secrecy conditions are considered, and secret sharing schemes that achieve these secrecy conditions are constructed. The \textit{weak secrecy} condition requires that each user does not get any information about the individual secrets of other users, while the \textit{perfect secrecy} condition requires that a user does not get any information about the collection of other users' secrets. Two major properties, namely, the storage overhead and the communication complexity, are defined for such distributed secret sharing systems. Optimal values for storage overhead and communication complexity were derived for any given $m$ and $n$, and protocols that achieve these optimal values simultaneously are constructed. 
The secret sharing protocols proposed in \cite{SM18} are specific to the case where each user has a designated secret message. In this paper, we consider the scenario where each user requests multiple secrets and propose protocols that achieve optimal storage overhead, ensuring weak secrecy. This can be seen as a generalization of the distributed multi-user secret system considered in \cite{SM18}. In \cite{KMM21}, the capacity region of the distributed multi-user secret sharing system under weak secrecy is characterized, where they consider the set-up in which each user can have a secret message of a different size. We generalize this result to the case where each user requests multiple secrets.

\textit{{Notation:}}
For $n\in \mathbb{N}$, define $[n]$ as the set $\{1,2,\hdots,n\}$ and for $n_1, n_2\in \mathbb{N}\cup \{0\}$, $n_1\leq n_2$, define $[n_1:n_2]$ as the set of $\{n_1,n_1+1,\hdots,n_2\}$. For a set $\mathcal{I}=\{i_1,\hdots,i_n\}$, $A_{\mathcal{I}}$ represents $\{A_{i_1},\hdots,A_{i_n}\}$. 
 
\vspace{-0.4cm}
\subsection{Our Contributions}
\vspace{-0.2cm}
In this paper, we consider a setting where each user requests a set of secrets and propose a secret sharing protocol that achieves optimal storage overhead under weak secrecy condition. The contribution and organization of this paper are as follows:
\begin{itemize}
    \item We derive a necessary condition on the storage structure of the distributed secret sharing protocol to ensure weak secrecy and establish a relation between the storage structure and the $t$-disjunct matrices. (Please see Section \ref{section:DSSPoptimalSO}, Lemma \ref{lemma:necessaryconditionsweaksecrecy}). 
    \item Using the storage structure obtained from the $t$-disjunct matrix, we propose a secret sharing protocol that achieves optimal storage overhead. (Please see Section \ref{subsection:DSSPSO}).
    \item Using several constructions for $t$-disjunct matrices, we compare the system parameters and properties. 
    We also show that $t$-disjunct matrices obtained using the Steiner system are better than those obtained from other known constructions in terms of accommodating more secrets. (Please see Section \ref{sec:wellknownconstructions}).
    \item For the DMUSS system considered in our problem, we provide a range in which the communication complexity lies and derive bounds on the optimal communication complexity. (Please see Section \ref{sec:commcomplexity}).
    \item We characterize the capacity region of the distributed multi-user secret sharing system when the storage structure is specified. (Please see Section \ref{sec:capacity}, Theorem \ref{lemma:DMUMSS}).
\end{itemize}

\section{System Model}
\subsection{System Model}
We consider a distributed secret sharing system that comprises $n$ storage nodes, $m$ ($m\geq n$) secrets, and $P = \binom{m}{t}$ users (Fig. \ref{fig:sim}), with the primary goal of enabling the dealer to convey a specific set of secrets to each user securely via storage nodes. In this system model:
\begin{itemize}
    \item Each secret $s_j$ has a \textit{storage set} $A_j \subseteq [n]$, which represents the set of all storage nodes that store the shares corresponding to secret $s_j$. For each $i \in A_j$, a share corresponding to secret $s_j$ is stored in $i$-th storage node. The set of all these storage sets is called the \textit{storage structure}, and it is denoted as 
    \vspace{-0.05in}
    \begin{align}
        \label{eq:accessstructure}
        \mathcal{A} \triangleq \{ A_j: j \in [m] \}.
    \end{align}
    \vspace{-0.25in}
    \item Storage nodes are passive, which means they do not communicate with each other. The users do not communicate with each other either. 
    \item All the secrets $s_j$, $j \in [m]$, are uniformly distributed and mutually independent. Each user $u$ requests a subset $S_u$ of $[m]$ secrets, where $|S_u| = t$.
    \item The dealer has access to all storage nodes but has no access to the users.
\end{itemize}
\begin{figure}[htbp]
  \centering
  \includegraphics[width=0.4\textwidth]{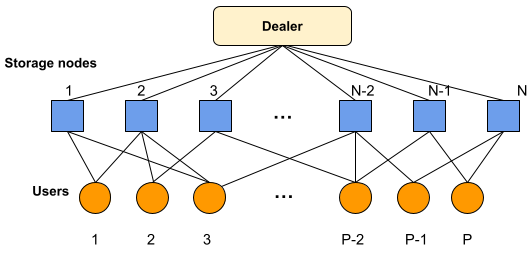}
  \setlength{\belowcaptionskip}{10pt}
  \caption{System Model}
  \label{fig:sim}
  \vspace{-5mm}
\end{figure}
 The aim is to develop a distributed secret sharing protocol that encodes secrets into shares and distributes them among storage nodes so that each user $u$ can successfully reconstruct their designated set of $t$ secrets and the secrecy condition is satisfied in a weak sense as defined below.
\begin{definition}
\label{defn:DSSP}
A distributed secret sharing protocol (DSSP) is a bundle of ($\mathcal{A}$, $\mathcal{E}$, $\bs{Z}_{n \times h}$, $\mathcal{D}$), where
\end{definition}
\begin{itemize}
    \item $\mathcal{A}$ is the storage structure defined in equation (\ref{eq:accessstructure}).
    \item $\mathcal{E}:\mathbb{F}_q^{m} \rightarrow \mathbb{F}_q^{h}$ with $h \geq m$ is an encoding function which relates to storage overhead of the system. The input $\bs{s} = (s_1,s_2,\hdots,s_m)^T$ is a vector of all secrets. The output $\bs{\mathrm{y}} = \mathcal{E}(\bs{s}) = (y_1,\hdots,y_h)^T$ is a vector of all data (shares) to be distributed and stored in the storage nodes. 
    \item $\bs{Z} = [z_{i,\scriptr}]_{n \times h}$, where $z_{i,\scriptr} = 1$ if $y_\scriptr$ is stored in $i^{th}$-storage node, otherwise $0$. We denote by $\bs{\mathrm{y}}_j$ the vector of all shares stored in nodes indexed by the elements of storage set $A_j$. The matrix $\bs{Z}$ is referred to as \textit{storing matrix}. The mapping of the output symbols to the storage nodes (specifying which output symbols are stored at each storage node) is referred to as the \textit{storage profile}. 
    \item $\mathcal{D}$ is a collection of $m$ decoding functions $\mathcal{D}_j:\mathbb{F}_q^{|\bs{\mathrm{y}}_j|} \rightarrow \mathbb{F}_q$, such that $\mathcal{D}_j(\bs{\mathrm{y}}_j) = s_j$. In other words, each user can successfully reconstruct its secrets. This is referred to as \textit{correctness} condition.
\end{itemize}

In the protocol, the \textit{weak secrecy} condition requires that a user does not get any information, in an information-theoretic sense, about the individual secrets of any other user. Let $U_j$ denote the set of all the data the user $j$ has access to, $S_{j}$ is the set of secrets requested by user $j$. Then,
\vspace{-0.2cm}
\begin{equation}
\label{eq:weaksecrecy}
    \forall j \in \left[\binom{m}{t}\right], \ \ell \in [m] \setminus S_{j}:\; H(s_\ell|U_j)=H(s_\ell).
\end{equation}

A DSSP satisfying the weak secrecy condition is also called a weakly secure DSSP.
The notions of $\textit{storage overhead}$ and $\textit{communication complexity}$ as defined in \cite{SM18} are recalled below. These are used throughout the paper to evaluate the efficiency of proposed DSSPs.
Note that the total number of $\mathbb{F}_q$-symbols stored in storage nodes is
$k' = \sum_{i=1}^{n} \sum_{r=1}^{h} z_{i,\scriptr}$, 
 where $\bs{Z} = [z_{i,\scriptr}]_{n \times h}$ is specified in Definition \ref{defn:DSSP}.
\begin{definition}
\label{defn:storageoverhead}
    The storage overhead, SO, of the DSSP is defined as
    \vspace{-0.07in}
    \begin{equation}
    \label{eq:storageoverhead}
        SO \triangleq \frac{k'}{m}.
    \end{equation}
\end{definition}
\vspace{-0.06in}
Note that the correctness condition must be satisfied for $m$ uniformly distributed and mutually independent secrets. Therefore, $k' \geq m$ and, consequently, $SO \geq 1$.
\begin{definition}
\label{defn:communicationcomplexity}
    Let $c_u$ denote the number of symbols user $u$ needs to download from the storage nodes to reconstruct the designated set of secrets $S_u$. Then the communication complexity $C$ is defined as
    \vspace{-0.5cm}
    {\small \begin{equation}
    \label{eq:CommComplexity}
        C \triangleq \sum_{u=1}^{\binom{m}{t}} c_u.
    \end{equation}}
\end{definition}
\subsection{Shamir's Secret Sharing Scheme}
Now, we describe the $(k,r)$-secret sharing scheme proposed by Shamir. Given a secret $s \in \mathbb{F}_q$, the output of the scheme consists of $k$ secret shares $d_1,d_2,\ldots,d_k \in \mathbb{F}_q$ satisfying the following conditions:
(i) Given $r$ or more secret shares, it is possible to reconstruct the secret $s$. (ii) In the information-theoretic sense, the knowledge of $r-1$ or fewer shares does not disclose any information regarding the secret $s$. Consider a $(r-1)$-degree polynomial $P(x)$ given by 
$ P(x) = s + \sum_{i=1}^{r-1}p_ix^i $, where $p_i$'s are i.i.d and are selected uniformly at random from $\mathbb{F}_q$. Let $\gamma_1,\gamma_2,\hdots,\gamma_k$ denote $k$ distinct non-zero elements from $\mathbb{F}_q$. The secret shares are then constructed by evaluating $P(x)$ at $\gamma_i$'s, i.e.,
$d_i = P(\gamma_i), \forall\ i \in [k].$ This is called $(k,r)$-Shamir's encoder.
Given any $r$ secret shares, $P(x)$ can be interpolated and is uniquely determined since the degree of $P(x)$ is at most $r-1$. This is called $(k,r)$-Shamir's decoder. 


\section{DSSP with Optimal Storage Overhead}
\label{section:DSSPoptimalSO}
\label{sec:DSSPoptimalSO}
In this section, we give a necessary condition on storage sets in a DSSP with weak secrecy, which relates to the disjunct matrices majorly used in the \textit{group testing} \cite{KS06, PR11, IKO20}. We then propose a scheme for constructing DSSPs with optimal storage overhead using the storage structure obtained from disjunct matrices.

\subsection{Conditions on storage sets to ensure weak secrecy}
\begin{lemma}
\label{lemma:necessaryconditionsweaksecrecy}
For any weakly secure DSSP with a storage structure
$\mathcal{A}$ defined in \eqref{eq:accessstructure}, we have
$A_{j_{t+1}}\nsubseteq \bigcup _{k=1}^t A_{j_k}$,  
$\forall\ j_1,j_2,\ldots,j_{t+1}\in [m]$ with $j_1\neq j_2\neq \cdots \neq j_{t+1}$.
\end{lemma}
\begin{proof}
Assume to the contrary that $A_{j_{t+1}}\subseteq (A_{j_1} \cup A_{j_2}\cup \cdots\cup A_{j_t})$ for some $j_1\neq j_2\neq \cdots\neq j_{t+1}$. This means that the user who has access to the secrets $s_{j_1},\hdots,s_{j_t}$, also has access to the secret $s_{j_{t+1}}$. This implies that the weak secrecy condition in (\ref{eq:weaksecrecy}) is violated, which is a contradiction.
\end{proof}

The collection of subsets satisfying the Lemma \ref{lemma:necessaryconditionsweaksecrecy} can be related to the columns of \textit{t-disjunct} matrices.
\begin{definition}
\label{defn:tdisjunctmatrix}
A $n \times m$ binary matrix $A$ is $t$-disjunct if the union of supports of any $t$ columns does not contain the support of any other column.
\end{definition}

Some well-known constructions for $t$-disjunct matrices are described in Section \ref{sec:wellknownconstructions}. We make the correspondence between the storage sets and $t$-disjunct matrices as follows: Consider a $t$-disjunct matrix where the columns correspond to the secrets, the rows correspond to the storage nodes, and the support of each column corresponds to the storage set of each secret.

\subsection{DSSP with Optimal Storage Overhead}
\label{subsection:DSSPSO}

Consider a system with $m$ secrets, and each user wants to access a subset of $t$ secrets, $t<m$. There can be at most $\binom{m}{t}$ users in the system. Let $\mathcal{A} = \{A_i:i\in[m]\}$ be the storage structure which consists of $m$ subsets, each corresponding to the storage sets of $m$ secrets. Suppose a user requests $\mathcal{P}\subset [m]$ secrets, then the user is given access to all the nodes in $\bigcup_{i\in \mathcal{P}}A_i$. To ensure weak secrecy, the storage structure $\mathcal{A}$ must satisfy the condition specified in Lemma \ref{lemma:necessaryconditionsweaksecrecy}. So, we consider $\mathcal{A}$ to be a set of supports of each column of a $t$-disjunct matrix. We consider $t$-disjunct matrices to have the same column weight (i.e., the same number of non-zero positions in each column). Therefore each storage set is of the same size. Consider, $|A_i| = r$, $A_i = \{n_{i,1},n_{i,2},\hdots,n_{i,r} \}, \forall i \in [m].$

For the purpose of decoding, Shamir's decoder is used. 
To initialize the protocol, we pick $n$ secrets such that the union of their storage sets is $[n]$. As each storage node is present in at least one of the storage sets in the storage structure considered, a set of $n$ such secrets always exists. Otherwise, we can ensure weak secrecy using a lesser number of storage nodes. WLOG, let these be the first $n$ secrets. To encode the secrets $s_j$, $\forall \ j\in [n]$ their $(r-1)$-degree polynomials $P_j(x)$s in (\ref{eq:polyeval}) are constructed by the following system of linear equations:
\vspace{-0.2cm}
\begin{equation}\label{eq:polyeval}
    P_{j}(x) =  s_{j}+ \sum_{l=1}^{r-1}p_{j,l}x^{l}
\end{equation}

\begin{center}
\label{eq:linearsystemofeqs}
\begin{tabular}{c|c|c}
$P_1(\gamma_1)=\alpha_{n_{1,1}} y_{n_{1,1}}$ &  $\cdots$ & $P_n(\gamma_1)=\alpha_{n_{n,1}}y_{n_{n,1}}$\\

 $P_1(\gamma_2)=\alpha_{n_{1,2}}y_{n_{1,2}}$ & $\cdots$ & $P_n(\gamma_2)=\alpha_{n_{n,2}}y_{n_{n,2}}$\\
\vdots & \vdots & \vdots \\
$P_1(\gamma_r)=\alpha_{n_{1,r}}y_{n_{1,r}}$ & $\cdots$ & $P_n(\gamma_r)=\alpha_{n_{n,r}}y_{n_{n,r}}$\\
\end{tabular}
\captionof{table}{Linear system of equations}\label{linearsystemofeqs}
\end{center}

\noindent As $\bigcup_{i=1}^n A_i = [n]$, we overlap the shares so that $\{y_{n_{1,1}},\hdots,y_{n_{1,r}},\hdots,y_{n_{n,1}}\hdots,y_{n_{n,r}}\}=y_{[1:n]}$, we have $nr$ variables and $nr$ equations. Using Lemma 1 in \cite{KMM21}, there exist some $\alpha_{n_{i,k}}$s and $\gamma_k$s such that the system has a unique solution for $p_{j,l}$s and $y_{n_{i,k}}$s, $i,j\in [n], l\in [r-1], k\in [r]$. Hence, there is a one-to-one mapping between $y_{[1:n]}$ and $s_{[1:n]}$.

\begin{algorithm}[H]
\centering
\caption{Proposed DSSP : Encoder}\label{alg:proposed_dssp}
\begin{enumerate}
    \item[(a)] \textbf{(Initialization)} Pick $n$ secrets such that the union of their storage sets is $[n]$.
    \item[(b)] \textbf{(Share Distribution for $\bs{n}$ secrets)} Pick distinct, non-zero $\gamma_1,\ldots,\gamma_{r}\in \mathbb{F}_q$ which are made public. For each $i\in [n]$, use encoder in Table \ref{linearsystemofeqs} to encode secret $s_i$ into shares $(y_{n_{i,1}}, \ldots, y_{n_{i,r}})$. Each share $y_i$ is stored in $i$-th storage node, where, by construction, we have $\bigcup_{i\in [n]}\{y_{n_{i,1}}, \ldots, y_{n_{i,r}}\} = \{y_1,\ldots, y_n\}$.
    \vspace{0.5em}
    \item[(c)] \textbf{(Share Distribution for remaining $\bs{m-n}$ secrets)} For each $i\in [n+1:m],$ 
    find a polynomial $P_i(x) = s_{i} + \sum_{j = 1}^{r-1}p_{i,j} x^{j}$ satisfying $P_i(\gamma_j) = y_{n_{i,j}}$ for each $j\in [r-1]$. The $r^{th}$ share for secret $s_{i}$ is $y_{n_{i,r}} = P_i(\gamma_{r})$ which is stored in $n_{i,{r}}$-th storage node.
\end{enumerate}
\end{algorithm}

The encoding of the remaining $m-n$ secrets is the same as the one proposed in \cite{SM18} ($t=1$ case) to achieve optimal storage overhead. 
The data symbols $y_{[1:n]}$ correspond to the shares of the first $n$ secrets, and $y_{[n+1:m]}$ correspond to the shares of remaining $m-n$ secrets.
%
%
The idea is to encode $n$-secrets into $n$-data symbols and then utilize them as the random seed required to encode the remaining $m-n$ secrets.
The lemma below proves that the proposed DSSP is indeed a weakly secure DSSP.
\begin{lemma}
    The proposed protocol is a weakly secure DSSP satisfying all the conditions in Definition \ref{defn:DSSP}.
\end{lemma}
\begin{proof}
    In this protocol, each user $j$ has access to all the $\sum_{i\in S_{j}}|A_i|$ evaluations of the polynomials associated with the secrets the user has requested. Hence, the correctness condition is satisfied by invoking Shamir's decoder. Now, we show that the proposed DSSP is indeed a weakly secure DSSP by showing that the condition specified in \eqref{eq:weaksecrecy} holds. 
    First, we show that the data symbols $y_1,y_2,\hdots,y_m$ generated according to the proposed protocol are uniformly distributed and mutually independent. As the vector of all secrets is assumed to be full entropy, ($s_1,s_2,\hdots,s_n$) is also full entropy. Also, under certain conditions (Lemma 1 in \cite{KMM21}), there is a one-to-one mapping between $y_{[1:n]}$ and $s_{[1:n]}$. This implies ($y_1,y_2,\hdots,y_n$) is also full entropy. Then,
    \begin{align}
        H(y_{[n+1:m]} |y_{[1:n]}) &= H(s_{[n+1:m]} |y_{[1:n]})\hspace{1.9cm} \label{eq:fullentropy1}
    \end{align}
    \begin{equation}
      \begin{minipage}{.8\linewidth}
        \hspace{2.3cm}
        \begin{eqlist}
          \item\label{eq:fullentropy2} $H(s_{[n+1:m]})$
          \item\label{eq:fullentropy3} $(m-n)\log q$,
        \end{eqlist}            
      \end{minipage}
    \end{equation}
    where (\ref{eq:fullentropy1}) holds since, given $y_{[1:n]}$ there is a one-on-one mapping between $y_{[n+1:m]}$ and $s_{[n+1:m]}$\cite{SM18},  (\ref{eq:fullentropy2}) holds since $y_{[1:n]}$ is independent of $s_{[n+1:m]}$ and (\ref{eq:fullentropy3}) holds since it is assumed that the vector of all secrets is full entropy. Using this together with the chain rule, we have
    \vspace{0.2cm}
    \begin{align}
        H(y_{[1:m]}) &= H(y_{[1:n]})+H(y_{[n+1:m]} |y_{[1:n]}) \nonumber\\
        &= n\log q + (m-n)\log q = m\log q. \label{eq:datasymbolsfullentropy}
    \end{align}
    Hence, from \eqref{eq:datasymbolsfullentropy}, we can say that the data symbols have full entropy and are mutually independent. As the storage sets are assumed to satisfy the $t$-disjunctness condition, there exists at least one $\gamma_{i}$, $i \in [r]$ such that $P_{\ell}(\gamma_{i})$ is not accessed by user $j$. Let this data symbol $P_{\ell}(\gamma_{i})$ be denoted by $y_{\ell}^{(-j)}$. Then $\forall\ j \in \left[\binom{m}{t}\right], \ell \in [m] \setminus S_{j}$
    \begin{align}
    \label{eq:entropyWS1}
        H(s_\ell|U_j) \geq H(s_\ell| \bs{\mathrm{y}}\setminus y_{\ell}^{(-j)})\hspace{3.9cm}
    \end{align}
    \begin{equation}
      \begin{minipage}{.8\linewidth}
        \hspace{0.8cm}
        \begin{eqlist}
          \item\label{eq:entropyWS2} $H(y_{\ell}^{(-j)}|\bs{\mathrm{y}}\setminus y_{\ell}^{(-j)})$
          \item\label{eq:entropyWS3} $H(y_{\ell}^{(-j)})$
          \item\label{eq:entropyWS4} $\log q.$
        \end{eqlist}           
      \end{minipage}
    \end{equation}
    where (\ref{eq:entropyWS1}) holds since conditioning does not increase the entropy, (\ref{eq:entropyWS2}) holds because given any $r-1$ evaluations of $P_{\ell}$ which is the evaluation polynomial corresponding to the $\ell$-{th} user, out of $r$ available ones, there is a one-to-one mapping between the remaining evaluation of $P_{\ell}$ and $s_{\ell}$. (\ref{eq:entropyWS3}) because data symbols are independent and (\ref{eq:entropyWS4}) because data symbols are full entropy. Also, we have
    \begin{align}
    \label{eq:entropyWS5}
          H(s_\ell|U_j) \leq H(s_\ell) = \log q,
    \end{align}
    $\forall j \in \left[\binom{m}{t}\right], \ell \in [m] \setminus S_{j}$.    From (\ref{eq:entropyWS4}) and (\ref{eq:entropyWS5}), the weak secrecy condition \eqref{eq:weaksecrecy} is satisfied.
\end{proof}

The encoding complexity of the proposed DSSP is $O(n^2 + rm)$. The encoding latency is $O(n^2)$. The decoding complexity is $O\left(rt \times \binom{m}{t}\right)$ for all users to decode their $t$ secrets, and the decoding latency is $O(rt)$. We refer the reader to \cite{SM18} for a detailed understanding of decoding and encoding complexities.

The total number of shares stored across all the storage nodes is equal to the total number of secrets, which implies that the storage overhead is 1. Hence, the proposed protocol has an optimal storage overhead. The proposed protocol also applies to cases where users request a different number of secrets. In that case, the value of $t$ is equal to the maximum over the number of secrets requested by each user.

\section{Comparison between different constructions of disjunct matrices}\label{sec:wellknownconstructions}

We first define some of the well-known constructions for $t$-disjunct matrices and then compare the number of storage nodes $n$ and communication complexity $C$ when each of these constructions is used as the storage structure in the DSSP described in Section \ref{sec:DSSPoptimalSO}. 

\subsubsection[Kautz-Singleton Construction]{Kautz-Singleton Construction \cite{KS06}}
 A $[q,k,q-k+1]_q$ Reed-Solomon code is picked as the outer code $\mathcal{C}_{out}$ while the inner code $\mathcal{C}_{in} : \mathbb{F}_q \rightarrow \{0,1\}^q$ is defined as follows. For any $i \in \mathbb{F}_{q}$, $\mathcal{C}_{in}(i) = e_{i}$, where $e_i$ is nothing but a one-hot vector. The concatenated code $\mathcal{C}^{*}= \mathcal{C}_{out} \circ \mathcal{C}_{in}$ is a $n \times m$ $t$-disjunct matrix, where $n = q^2, m = q^k,$ and $t = \lfloor\frac{q-1}{k-1}\rfloor$. Here, each column in $\mathcal{C}^*$ has $q$ ones. The set of storage sets obtained from the Kautz-Singleton construction is called the Kautz-Singleton storage structure.

\subsubsection[Porat-Rothschild Construction]{Porat-Rothschild Construction \cite{PR11}}
A linear code that meets the Gilbert-Varshamov bound is picked as the outer code. Here, $\mathcal{C}_{out}$ is $[\scriptr, k, \delta \scriptr]_{q}$ linear code,
where $\scriptr{\leq} \frac{k}{1-H_q(\delta)}$ and $t+1 = \lceil \frac{1}{1-\delta}\rceil$. The inner code $\mathcal{C}_{in} : \mathbb{F}_q \rightarrow \{0,1\}^q$ is defined as follows. For any $i \in \mathbb{F}_{q}$, $\mathcal{C}_{in}(i) = e_{i}$, where $e_i$ is nothing but a one-hot vector. The concatenated code $\mathcal{C}^{*} = \mathcal{C}_{out} \circ \mathcal{C}_{in}$ is a $n\times m$ $t$-disjunct matrix, where $n=\scriptr q, m=q^k$ and $t = \lceil \frac{1}{1-\delta}\rceil - 1$. Here, each column in $\mathcal{C}^{*}$ has $\scriptr$ ones. The set of storage sets obtained from Porat-Rothschild Construction is called the Porat-Rothschild storage structure.

\subsubsection[Sparse Disjunct Matrices]{Sparse Disjunct Matrices \cite{IKO20}}
In a $n\times m$ Sparse disjunct matrix, the number of ones in each column of a $t$-disjunct matrix is restricted to $\ell t+1$, where $\ell \geq 1$. Such a matrix can be constructed by replacing the outer code in the Kautz-Singleton construction with $[\ell t+1,k= \ell +1]$-RS code over a field of size $q = \sqrt[\ell + 1]{m}$. The concatenated code $\mathcal{C}^{*}$ is a $(\ell t+1)q\times m$ $t$-disjunct matrix. The set of storage sets obtained from constructing the Sparse Disjunct Matrix is called Sparse Disjunct storage structure.

\begin{figure*}
\centering
{\small
    \begin{tabular}{|c|c|c|c|}
    \hline
         & Kautz-Singleton & Porat-Rothschild & Sparse disjunct ($\ell = 1$) \\
    \hline
         & $\mathcal{C}_{out}: [q_1, k, q_1-k+1]_{q_1}$- RS code  & $\mathcal{C}_{out}:[\scriptr, k, \delta \scriptr]_{q_1}$-Linear code & $\mathcal{C}_{out}:[t + 1, k = 2]_{q_2}$-RS code \\
         
         Code & & where $\scriptr\leq \frac{k}{1-H_{q_1}(\delta)}$ & where $q_2=\sqrt{m}$ \\
        
         & $\mathcal{C}_{in}: I_{q_1}$ & $\mathcal{C}_{in}: I_{q_1}$ & $\mathcal{C}_{in}: I_{q_2}$ \\
    \hline
        Disjunctness & $t = \lfloor \frac{q_1-1}{k-1} \rfloor$ & $t+1 = \lceil \frac{1}{1- \delta} \rceil $ & $t+1\leq \sqrt{m}$ \\
    \hline
        $t$-disjunct Matrix & $\mathcal{C}^*:q_1^2\times m$ & $\mathcal{C}^*:\scriptr q_1 \times m$ & $\mathcal{C}^*: (t+1)q_2\times m$ \\
         & $\approx t^2 \log_{t}^2m \times m$ & $\approx t^2 \log m \times m$ & $= (t+1)\sqrt{m}\times m$ \\
    \hline
     Column weight & $q_1\approx t \log_{t}m$ & $\scriptr$ & $t+1$ \\
     \hline 
     Storage Overhead & 1 & 1 & 1 \\
     \hline
     Comm. Complexity & $\binom{m - 1}{t-1} \times q_1 \times m$& $\binom{m - 1}{t-1} \times \scriptr \times m$ & $\binom{m-1}{t-1} \times (t + 1) \times m$ \\ 
     \hline
    \end{tabular}
    \captionof{table}{Comparison of disjunct matrix constructions.}\label{compare3}
    }
\end{figure*}
For a given total number of secrets $m$, and the number of secrets $t$ each user requests $(t<<m)$, we compare the number of storage nodes $n$ and the communication complexity $C$ across the above mentioned constructions for $t$-disjunct matrices (see also Table \ref{compare3}).

\subsubsection*{Kautz-Singleton\,(KS) Vs Porat-Rothschild\,(PR)}
There are two regimes considered in the literature: (i) $t = O(poly(\log m))$, and (ii) $t=O(m^\alpha),\ \alpha\in (0,1/2)$. In the regime (i), we have $n_{PR}<n_{KS}$ and $C_{PR}<C_{KS}$, which shows that PR is better than KS. Whereas in the regime (ii), both the inequalities are reversed, which shows that KS is better than PR.


\subsubsection*{Kautz-Singleton\,(KS) Vs Sparse Disjunct ($\ell = 1$)\,(SD)}
Since we have $t<< m$, $n_{KS}<n_{SD}$ and $C_{KS}>C_{SD}$. There is a tradeoff between these storage structures, so if for example one seeks to minimize the number of storage nodes, KS is better than SD while paying for a higher communication complexity and vice versa.

\subsubsection*{Porat-Rothschild\,(PR) vs Sparse disjunct ($\ell =1$)\,(SD)}
Since we have $t<< m$, $n_{PR}<n_{SD}$ and $C_{PR}<C_{SD}$. Thus, PR is better than SD.



Another interesting construction of the $t$-disjunct matrix uses the Steiner system, defined below.

\begin{definition}
Let $X$ be an $n$-element set. A Steiner system $\mathfrak{S}(n,b,p)$ is defined as $\mathfrak{S}\subset \binom{X}{b}$ such that for every $A \in \binom{X}{p}$ there is exactly one $B\in \mathfrak{S}$ with $A \subset B$, where $\binom{X}{i}$ here denotes the collection of all the $i$-sized subsets of $X$. The largest set $\mathfrak{S}$ which satisfies this property is called the \textit{maximum} Steiner system.
\end{definition}

In \cite{steinerconjecture}, it is proved that $\mathfrak{S}(n,b,p)$ Steiner system gives a $\lfloor \tfrac{b-1}{p-1}\rfloor$ disjunt matrix. A relation between the Steiner system and constant column weight $t$-disjunct matrices was conjectured, as stated below.

\begin{conjecture}{\cite{steinerconjecture}}\label{conjecture:steinersystem}
    Let $M$ be a $n\times m$ $t$-disjunct matrix with constant column weight $b$. Let $p = \tfrac{b+t-1}{t}$ (we assume that $p$ is an integer). 
    The maximum $\mathfrak{S}(n,b,p)$ Steiner system gives a matrix $M^\prime$ that is no worse than $M$. In other words, $|\mathfrak{S}(n,b,p)|\ge m$.
\end{conjecture}

\begin{remark}
    If Conjecture \ref{conjecture:steinersystem} is true, then for a given $n$, the number of storage nodes, and $b$, the size of storage sets, the storage structure obtained from Steiner system $\mathfrak{S}(n,b,p)$ is the best in terms of accommodating more number of secrets.
\end{remark}


We prove this conjecture for the special cases where the matrix $M$ is obtained using the Kautz-Singleton and Sparse Disjunct constructions.

Consider a $q^2\times q^k$ $t$-disjunct matrix obtained from Kautz-Singleton construction, $t = \lfloor \frac{q-1}{k-1}\rfloor$. We set $k = \frac{q+t-1}{t}$ to accommodate most secrets in this construction. Then the corresponding Steiner system with the same number of storage nodes is given by $\mathfrak{S}(q^2,q,k)$. The following lemma compares the total number of secrets accommodated by both constructions. 

\begin{lemma}
    If the Steiner system $\mathfrak{S}(q^2,q,p=\tfrac{q+t-1}{t})$ exists, then it can accommodate more secrets compared to the storage structure obtained from the Kautz-Singleton construction with the same number of nodes $q^2$ and constant column weight $q$. In other words,
    \begin{align}
        q^{s} &\leq \left|\mathfrak{S}(q^2,q,p)\right|
        =  \binom{q^2}{s}\Big/\binom{q}{s}.\label{eqn:steiner_ks}
    \end{align}
\end{lemma}

\begin{proof}
    Expanding the binomial coefficients on RHS of \eqref{eqn:steiner_ks} and observing that for all $\ell\in [1,q),$ $\tfrac{q^2-\ell}{q-\ell}>q$ gives \eqref{eqn:steiner_ks}.
\end{proof}

Similarly, a comparison between a sparse disjunct matrix and the Steiner system is given below.

\begin{lemma}
    If the Steiner system $\mathfrak{S}((t+1)q,t+1,2)$ exists, then it can accommodate more secrets compared to the storage structure obtained from the Sparse Disjunct matrix with the same number of nodes $(t+1)q$ and constant column weight $t+1$. In other words,
    \begin{align*}
        q^{2} &\leq \left|\mathfrak{S}((t+1)q,t+1,2)\right|
        =  \binom{(t+1)q}{2}\Big/\binom{t+1}{2}.
    \end{align*}
\end{lemma}

\subsubsection*{\textbf{Balanced Storage Profile}} In large-scale distributed storage systems, it is essential to distribute the data evenly across the nodes and ensure each node has a similar amount of data to manage. This helps in avoiding problems like slower access times and system failures. We say that a collection $\mathcal{F}$ of subsets of $[n]$ is a \textit{balanced collection} if each $i \in [n]$ belongs to the same number of subsets in $\mathcal{F}$. The Kautz-Singleton construction and the Sparse disjunct matrices defined above provide a $t$-disjunct matrix with constant row and column weights. Thus, the resulting storage structures are balanced collections and can be used to obtain a DSSP with a balanced storage profile.

\section{Bounds on Optimal Communication Complexity}\label{sec:commcomplexity}
In this section, we derive bounds on the minimum communication complexity of DSSPs where each user requests a subset of $t$ secrets. Similar to \cite{SM18}, we also use the \textit{tight} DSSPs in deriving these bounds. The DSSP which attains the lower bound of the minimum communication complexity with equality is called \textit{communication-optimal} DSSP.

\begin{definition}
    A DSSP is said to be tight DSSP (T-DSSP) if every user downloads exactly one $\mathbb{F}_q$-symbol from each node in the storage set corresponding to each secret in his designated set of secrets.
\end{definition}

Let $b_k$ denote the number of $t$ subsets whose union is of size $k$ in the storage structure of a T-DSSP. Then its communication complexity lies between:
\vspace{-0.1cm}
\begin{equation}
\label{eq:commcomplexity}
    \sum_{k=t}^{n}kb_{k} \leq C < t\sum_{k=t}^{n}kb_{k},
\end{equation}
where we obtain the lower and upper bounds when each user downloads exactly one share and $t$ shares, respectively, from each node, the user has access to.

In \cite{SM18}, it is proved that for every DSSP with communication complexity $C$, there exists a T-DSSP with the same number of storage nodes and users with communication complexity $\tilde{C} \leq C$. Therefore, we can minimize \eqref{eq:commcomplexity} to find communication-optimal DSSP, provided that storage structure satisfying Lemma \ref{lemma:necessaryconditionsweaksecrecy}, with such $b_{k}s$ exists. 

We derive a necessary condition for the storage structure satisfying Lemma \ref{lemma:necessaryconditionsweaksecrecy} to exist. 
\begin{lemma}\label{lemma:opt_CC_bound}
    Consider a storage structure $\mathcal{A}$ satisfying Lemma \ref{lemma:necessaryconditionsweaksecrecy}, then
    $\sum_{k=t}^n b_k/{\binom{n}{k}}\leq 1$.
\end{lemma}
\begin{proof}
The permutations of $[n]$ can be counted in two different ways using the double counting argument. One is by counting all permutations of $[n]$ identified with $\{1,\hdots, n \}$ directly, and there are $n!$ of them, and the other by generating a permutation of the $[n]$ by selecting sets $(S_{i_1},\hdots, S_{i_t}),$ each $S_{i_j}\in \mathcal{A}$ and choosing a map that sends $\{1,\hdots, |\cup_{j\in [t]} S_{i_j}|\}$ to $\cup_{j\in [t]} S_{i_j}$.

If $|\cup_{j\in [t]} S_{i_j}| = k$, the sets $(S_{i_1},\hdots, S_{i_t})$ are associated in this way with $k!(n - k)!$ permutations, and in each of them the image of first $k$ elements of $[n]$ is exactly $\cup_{j\in [t]} S_{i_j}$. Each permutation may only be associated with a single $\cup_{j\in [t]} S_{i_j}$. If a permutation is associated with $(S_{i_1},\hdots, S_{i_t})$ and $(S_{i^\prime_1},\hdots, S_{i^\prime_t})$, then one union would be a subset of the other. The number of permutations that this procedure can generate is less than or equal to $n!$, i.e.,
\begin{align*}
    \sum_{\substack{S_{i_j}\in A\\\forall j\in[t]}} |\cup_{j\in [t]} S_{i_j}|(1-|\cup_{j\in [t]} S_{i_j}|)=\sum_{k=t}^n b_k k!(n-k)! \leq n!.
\end{align*}
Dividing the above inequality by $n!$ gives the result. This is a generalization of the LYM inequality\cite{Lubell1966ASP}.
\end{proof}
From \eqref{eq:commcomplexity} and Lemma \ref{lemma:opt_CC_bound}, we consider the following discrete optimization problem to derive the bounds on optimal communication complexity.
\vspace{-0.1in}
\begin{align}
    \label{eq:objectivefn}
    min \quad &\sum_{k=t}^{n}kb_{k} \\
    \label{eq:constraint1}
    \textrm{s.t.} \quad &\forall k \in [t:n]: b_{k} \in \mathbb{N} \cup \{0\}\\
    \label{eq:constraint2}
         \quad &\sum_{k=t}^{n} b_{k} = {\binom{m}{t}} \\
    \label{eq:constraint3}
         \quad &\sum_{k=t}^{n} \frac{b_k}{\binom{n}{k}} \leq 1.
\end{align}
Constraint (\ref{eq:constraint1}) is set to ensure $b_{k}s$ are non-negative. Constraint (\ref{eq:constraint2}) is set because the sum of $b_{k}s$ is equal to the number of users $\binom{m}{t}$ and constraint (\ref{eq:constraint3}) is a necessary condition for storage structure satisfying Lemma \ref{lemma:necessaryconditionsweaksecrecy} to exist.

The solution to this constrained optimization problem follows the same lines as the one given in \cite{SM18}. Let $\beta_{k}^*$s be the solutions to the corresponding continuous optimization problem of the one given in (\ref{eq:objectivefn}), then at most two of the $\beta_{k}^*$s can be non-zero. Furthermore, if two are non-zero, their indices are consecutive. Let $i$ denote the largest integer such that $\binom{n}{i}\leq \binom{m}{t}$. Then,
\begin{align*}
    \beta^{*}_i = \frac{\binom{n}{i+1}-\binom{m}{t}}{\binom{n}{i+1}-\binom{n}{i}}\binom{n}{i}, \quad
    \beta^{*}_{i+1} = \frac{\binom{m}{t}-\binom{n}{i}}{\binom{n}{i+1}-\binom{n}{i}}\binom{n}{i+1}.
\end{align*}
For a given number of secrets $m$ and number of storage nodes $n$, any T-DSSP with a storage structure $\mathcal{A}$ that has $\lfloor \beta_i^{*} \rfloor$ $t$ subsets whose union is of size $i$ and $\lceil \beta_{i+1}^* \rceil$ $t$ subsets whose union is of size $i+1$ is a communication-optimal DSSP. We leave it as future work to solve for the exact value of optimal communication complexity.

\section{Capacity region of Distributed Multi-User Secret Sharing}
\label{sec:capacity}
In \cite{KMM21}, the capacity region of the distributed multi-user secret sharing system is characterized, subject to correctness and secrecy constraints. They consider a DMUSS system that consists of a dealer, $n\in \mathbb{N}$ storage nodes, and $m\in \mathbb{N}$ users. In the system set-up considered in \cite{SM18}, the storage structure has a regularized form, and the user's secret messages have equal size, whereas the DMUSS set-up in \cite{KMM21} considers an arbitrary storage structure, and the users can have different message sizes. Now, we recall a few definitions given in \cite{KMM21}:
\begin{definition}
    Let the length of secret $s_j$ be denoted by $r_j$; we define its secret message rate as $R_j = \frac{r_j}{K}$,
    where $K$ is the size of each storage node.
\end{definition}
\begin{definition}
    The \textit{capacity region} of a DMUSS is defined as a set of all achievable rate tuples, subject to the correctness and secrecy constraints.
\end{definition}
The capacity region of the distributed multi-user secret sharing system is characterized under weak secrecy condition \eqref{eq:weaksecrecy} as follows:
\begin{theorem}
\label{lemma:DMUMSS}
    The capacity region of DMUSS is the convex hull of all regions with the rate tuple $(R_1, R_2,\ldots,R_m)$ satisfying:
    \begin{gather}
        R_j \leq \min|A_j \setminus \cup_{\tilde{j} \in S}A_{\tilde{j}}|,\ \forall\  S \subseteq [m]\setminus\{j\}, |S|=t \\
        \sum_{i \in S} R_i \leq |\cup_{i \in S} A_i|,\quad S \subseteq [m]
    \end{gather}
\end{theorem}
The capacity region of DMUSS, characterized in \cite{KMM21}, is a special case of Theorem \ref{lemma:DMUMSS} where $t=1$. The achievability and converse proofs for the above theorem follow along the same lines as in \cite{KMM21}.


\bibliographystyle{IEEEtran} 
\bibliography{Bibliography} 
\end{document}